\documentclass[aps,onecolumn,nofootinbib,tightenlines,superscriptaddress]{revtex4}

\usepackage{amssymb,amstext,amsmath,amsthm,slashed}
\usepackage[dvips]{graphicx}
\usepackage{latexsym}
\usepackage{psfrag}
\usepackage{amsfonts}
\usepackage{color}
\usepackage{verbatim}
\usepackage{bbm}
\usepackage{dcolumn}
\usepackage{tabularx}
\usepackage{leftidx}
\usepackage{tensor}
\usepackage{boxedminipage}
\usepackage{fancyvrb}
\usepackage{float}
\usepackage{wrapfig}
\usepackage{booktabs}
\usepackage{multirow}
\usepackage{subfigure}
\usepackage{dsfont} 
\usepackage{mathrsfs}
\usepackage[hidelinks]{hyperref}
\usepackage{scalerel,stackengine}

\newtheorem{Def}{Definition}[section]

\newtheorem{Prp}[Def]{Proposition}

\begin{document}

\title{A Family of Horizon-penetrating Coordinate Systems for the Schwarzschild Black Hole Geometry with Cauchy Temporal Functions \vspace{0.1cm}}

\author{Christian R\"oken\let\thefootnote\relax\footnote{e-mail: croken@correo.ugr.es, croeken@uni-bonn.de}} 

\affiliation{Department of Geometry and Topology - Faculty of Science, University of Granada, 18071 Granada, Spain \vspace{0.25cm}}

\affiliation{Lichtenberg Group for History and Philosophy of Physics - Institute of Philosophy, University of Bonn, 53113 Bonn, Germany \vspace{0.25cm}}

\date{September 2020 / December 2025}

\begin{abstract}
\vspace{0.4cm} \noindent \textbf{\footnotesize ABSTRACT.} \, We introduce a new family of horizon-penetrating coordinate systems for the Schwarzschild black hole geometry featuring time coordinates that are Cauchy temporal functions for which the level sets are smooth, asymptotically flat, spacelike Cauchy hypersurfaces. Coordinate systems of this kind are well suited for the study of the temporal evolution of matter and radiation fields in the joined exterior and interior regions of the Schwarzschild black hole geometry, whereas the associated foliations can be employed as initial data sets for the globally hyperbolic development under the Einstein flow. For their construction, we formulate an explicit method that utilizes the geometry of---and structures inherent in---the Penrose diagram of the Schwarzschild black hole geometry, thus relying on the corresponding metrical product structure. As an example, we consider an integrated algebraic sigmoid function as the basis for the determination of such a coordinate system. Finally, we generalize our results to the Reissner--Nordstr\"om black hole geometry up to the Cauchy horizon. The geometric construction procedure presented here can be adapted to yield similar coordinate systems for various other spacetimes with the same metrical product structure.
\end{abstract}

\setcounter{tocdepth}{2}

\vspace{0.1cm}

\maketitle

\tableofcontents

\section{Introduction}

\noindent In a certain class of Lorentzian $4$-manifolds, there exist preferred $2$-dimensional, totally geodesic submanifolds with induced metrics that are locally conformally equivalent to the $2$-dimensional Minkowski metric, comprising central information on their global 4-dimensional geometries. Accordingly, these submanifolds may be used to faithfully represent and analyze the causal structures of the underlying Lorentzian $4$-manifolds. For first applications of this approach to the Schwarzschild, Reissner--Nordstr\"om, and Kerr geometries, we refer the reader to, e.g., \cite{Carter1, Carter2, Fink, GraBri, KRUS}. As can be seen, i.a., from these first applications, one of the most prominent examples of such preferred $2$-dimensional submanifolds are the $2$-surfaces containing the two double principal null directions in a Petrov type D solution of the vacuum Einstein field equations in general relativity \cite{Step, GrPo}, which can be employed to study the causal structures of black hole geometries, for instance, by means of Penrose diagrams \cite{CrOlSz, SchAgu}. In particular, this concept of analyzing the causal structures of Lorentzian $4$-manifolds is especially useful in the context of global hyperbolicity, which is a specific causality condition on Lorentzian manifolds that gives rise to foliations by smooth, spacelike Cauchy hypersurfaces. Thus, it is relevant for the initial value formulation of the Einstein field equations (see, e.g., \cite{ICB}), where one works with spacelike Cauchy hypersurfaces as initial data sets and derives solutions evolving this data forward and backward in time. Examples of such foliations of the maximal globally hyperbolic extensions of some spherically symmetric Lorentzian $4$-manifolds by spacelike Cauchy hypersurfaces that are maximal or have constant mean curvature can be found in \cite{BeigMurchadha, BrillCavalloIsenberg, EardleySmarr, EWCDST, ReimannBruegmann, Reinhart, SmarrYork, WaldIyer}.

In this work, we focus on an explicit construction procedure of global coordinate systems for a specific globally hyperbolic subset of the family of spherically symmetric vacuum geometries, namely the Schwarzschild black hole geometry, which are related to foliations of this geometry by smooth, asymptotically flat, spacelike Cauchy hypersurfaces. (The Schwarzschild black hole geometry is isometric to a subset of the maximally extended Schwarzschild geometry, viz.\ its regions I and II \cite{Wald}, and may be used in order to describe the final equilibrium state in the time evolution of the gravitational field of an isolated, nonrotating, uncharged black hole.) More precisely, we present a $2$-dimensional geometric construction procedure of a new family of horizon-penetrating coordinate systems with Cauchy temporal functions (Cauchy coordinates) covering the joined exterior and interior regions of the Schwarzschild black hole geometry, in which we deform the geometric shape of the corresponding Penrose diagram from a trapezoid into a centrally symmetric diamond via affine as well as homotopy transformations, and formulate conditions for the determination of families of smooth functions foliating this diamond. These functions are identified with smooth, spacelike Cauchy hypersurfaces in the Schwarzschild black hole geometry, which are asymptotically flat at spacelike infinity, encounter the curvature singularity only asymptotically, and yield regular foliations across the event horizon. Hence, the labels of these hypersurfaces are Cauchy temporal functions on the Schwarzschild black hole geometry, and may serve as time variables of the aforementioned global coordinate systems. (For the study of other families of horizon-penetrating coordinate systems related to foliations with similar boundary conditions and spatial slices with trumpet geometry, which, however, rely on different geometric construction procedures and do not, in general, yield foliations that cover the entire Schwarzschild black hole geometry up to the singularity, see \cite{DennisonBaumgarte, HHOBM}.) Having a coordinate system of this type at one's disposal may be advantageous in the derivation of propagators for matter and radiation fields in a Schwarzschild black hole background geometry in the framework of (relativistic) quantum theory. Moreover, the foliations associated with these coordinate systems can be used as initial data sets for the globally hyperbolic development of the Schwarzschild black hole geometry under the Einstein flow, tracing its evolution over time. 

The paper is organized as follows. In Section \ref{prel}, we first recall the main geometrical and topological aspects of the Schwarzschild black hole geometry, present a derivation of compactified Kruskal--Szekeres coordinates, and study the corresponding Penrose diagram. We then give a brief account of the notions of Cauchy surfaces and time-type functions. Subsequently, in Section \ref{III}, we introduce our geometric method for the explicit construction of horizon-penetrating Cauchy coordinate systems for the Schwarzschild black hole geometry. We also prove that the level sets of the time variables of these coordinate systems are in fact Cauchy hypersurfaces. The details of a specific example based on an integrated algebraic sigmoid function are worked out in Section \ref{IV}. In Section \ref{V}, we generalize our results to the Reissner--Nordstr\"om black hole geometry up to the Cauchy horizon. Finally, we conclude with a brief outlook on future research projects in Section \ref{VII}.

\section{Preliminaries} \label{prel}

\subsection{The Schwarzschild Black Hole Geometry and Compactified Kruskal--Szekeres Coordinates} \label{IA}

\noindent The Schwarzschild black hole geometry $(\mathfrak{M}, \boldsymbol{g})$ is a connected, smooth, globally hyperbolic and asymptotically flat Lorentzian $4$-manifold with $\mathfrak{M}$ being homeomorphic to $\mathbb{R}^2 \times S^2$ and a spherically symmetric metric $\boldsymbol{g}$, referred to as the Schwarzschild metric, which constitutes a $1$-parameter family of solutions of the vacuum Einstein field equations $\textnormal{Ric}(\boldsymbol{g}) = \boldsymbol{0}$. In the standard Schwarzschild coordinates $(t, r, \theta, \varphi) \in \mathbb{R} \times \mathbb{R}_{> 0} \times (0, \pi) \times [0, 2 \pi)$, this metric takes the form \cite{Schwarzschild}
\begin{equation} \label{BLmetric}
\boldsymbol{g} = \biggl[1 - \frac{2 M}{r}\biggr] \textnormal{d}t \otimes \textnormal{d}t - \biggl[1 - \frac{2 M}{r}\biggr]^{- 1} \textnormal{d}r \otimes \textnormal{d}r - r^2 \, \boldsymbol{g}_{S^2} \, ,
\end{equation}
where the parameter $M \in \mathbb{R}_{> 0}$ coincides with the ADM mass of the black hole geometry and $\boldsymbol{g}_{S^2} = \textnormal{d}\theta \otimes \textnormal{d}\theta + \sin^2{(\theta)} \, \textnormal{d}\varphi \otimes \textnormal{d}\varphi$ is the metric on the unit $2$-sphere. This representation of the Schwarzschild metric is well-defined for all $r \in \mathbb{R}_{> 0} \backslash \{2 M\}$ and features two types of singularities, namely a spacelike curvature singularity at $r = 0$ and a coordinate singularity at $r = 2 M$, with the latter being the location of the event horizon $\mathfrak{M} \cap \partial J^-(\mathscr{I}^+)$, that is, the boundary of the causal past of future null infinity. The Schwarzschild black hole geometry may thus be separated into two connected components: the component $\textnormal{B}_{\textnormal{I}} := \mathbb{R} \times \mathbb{R}_{> 2 M} \times S^2$, which is the domain of outer communication, and the component $\textnormal{B}_{\textnormal{II}} := \mathbb{R} \times (0, 2 M) \times S^2$, which is the future trapped region or black hole region $\mathfrak{M} \backslash J^-(\mathscr{I}^+) \not= \emptyset$. We remark that on $\textnormal{B}_{\textnormal{I}}$, the Schwarzschild time coordinate $t$ is a Cauchy temporal function, i.e., it yields a foliation of this region by smooth, spacelike Cauchy hypersurfaces (see Section \ref{IC}). However, due to the degeneracy of the Schwarzschild coordinates at---and the violation of the staticity of the Schwarzschild metric across---the event horizon, the level sets of $t$ do not foliate the entire Schwarzschild black hole geometry.

We next recall the usual derivation of compactified Kruskal--Szekeres coordinates, which we restrict, for the purposes of the present work, to the region $\textnormal{B}_{\textnormal{I}} \cup \textnormal{B}_{\textnormal{II}}$. These coordinates are single-valued and regular for all values of $r \in \mathbb{R}_{> 0}$, locate the event horizon at finite coordinate values, and give rise to a compactification of the total black hole geometry (required for the construction of Penrose diagrams). We begin by transforming the above Schwarzschild coordinates into Eddington--Finkelstein double-null coordinates \cite{Edd, Fink}
\begin{equation*} 
\mathfrak{T}^{\textnormal{EF}} \colon
\begin{cases}
\, \mathbb{R} \times \mathbb{R}_{> 0} \times (0, \pi) \times [0, 2 \pi) \rightarrow \mathbb{R} \times \mathbb{R} \times (0, \pi) \times [0, 2 \pi) \vspace{0.3cm} \\
\hspace{2.573cm} (t, r, \theta, \varphi) \mapsto (u, v, \theta', \varphi') 
\end{cases} 
\end{equation*}
with 
\begin{equation*} 
\left\{\!\begin{aligned}
& \, u = t - r_{\star} & & \textnormal{and} \quad v = t + r_{\star} & & \, \textnormal{for} \,\,\, \textnormal{B}_{\textnormal{I}} \\[0.1cm] 
& \, u = t + r_{\star} & & \textnormal{and} \quad v = - t + r_{\star} & & \, \textnormal{for} \,\,\, \textnormal{B}_{\textnormal{II}} 
\end{aligned}\right\} \, , \quad \theta' = \theta \, , \quad \textnormal{and} \quad \varphi' = \varphi \, ,
\end{equation*}
where
\begin{equation*} 
r_{\star} := r + 2 M \, \ln{\bigg|\frac{r}{2 M} - 1\bigg|} \in \begin{cases}
\, \mathbb{R} & \, \textnormal{for} \,\,\, \textnormal{B}_{\textnormal{I}} \\ 
\, \mathbb{R}_{< 0} & \, \textnormal{for} \,\,\, \textnormal{B}_{\textnormal{II}}
\end{cases}
\end{equation*}
is the Regge--Wheeler coordinate and $v + u \in \mathbb{R}_{< 0}$ for $\textnormal{B}_{\textnormal{II}}$. The Schwarzschild metric in Eddington--Finkelstein double-null coordinates reads
\begin{equation*}
\boldsymbol{g} = \frac{1}{2} \, \bigg|1 - \frac{2 M}{r}\bigg| \, (\textnormal{d}u \otimes \textnormal{d}v + \textnormal{d}v \otimes \textnormal{d}u) - r^2 \, \boldsymbol{g}_{S^2} \, .
\end{equation*}
We then apply the transformation into Kruskal--Szekeres double-null coordinates 
\begin{equation*} 
\mathfrak{T}^{\textnormal{KS1}} \colon
\begin{cases}
\, \displaystyle \mathbb{R} \times \mathbb{R} \times (0, \pi) \times [0, 2 \pi) \rightarrow \displaystyle \biggl(- \frac{\pi}{2}, \frac{\pi}{2}\biggr) \times \biggl(0, \frac{\pi}{2}\biggr) \times (0, \pi) \times [0, 2 \pi) \vspace{0.3cm} \\
\hspace{2.08cm} (u, v, \theta, \varphi) \mapsto (U, V, \theta', \varphi')
\end{cases} 
\end{equation*}
with 
\begin{equation*} 
\left\{\!\begin{aligned}
& \, \tan{(U)} = - e^{- u/(4 M)} & & \textnormal{and} \quad \tan{(V)} = e^{v/(4 M)} & & \, \textnormal{for} \,\,\, \textnormal{B}_{\textnormal{I}} \\[0.1cm] 
& \, \tan{(U)} = e^{u/(4 M)} & & \textnormal{and} \quad \tan{(V)} = e^{v/(4 M)} & & \, \textnormal{for} \,\,\, \textnormal{B}_{\textnormal{II}} 
\end{aligned}\right\} \, , \quad \theta' = \theta \, , \quad \textnormal{and} \quad \varphi' = \varphi \, ,
\end{equation*}
where $U \in (- \pi/2, 0)$ and $V \in (0, \pi/2)$ for $\textnormal{B}_{\textnormal{I}}$ and $U \in (0, \pi/2 - V)$ and $V \in (0, \pi/2)$ for $\textnormal{B}_{\textnormal{II}}$. Finally, we transform the Kruskal--Szekeres double-null coordinates into a compactified form of the usual Kruskal--Szekeres spacetime coordinates \cite{KRUS, Szekeres} 
\begin{equation*} 
\mathfrak{T}^{\textnormal{KS2}} \colon
\begin{cases}
\, \displaystyle \biggl(- \frac{\pi}{2}, \frac{\pi}{2}\biggr) \times \biggl(0, \frac{\pi}{2}\biggr) \times (0, \pi) \times [0, 2 \pi)  \rightarrow \displaystyle \biggl(- \frac{\pi}{4}, \frac{\pi}{4}\biggr) \times \biggl(- \frac{\pi}{4}, \frac{\pi}{2}\biggr) \times (0, \pi) \times [0, 2 \pi) \vspace{0.3cm} \\
\hspace{4.18cm} (U, V, \theta, \varphi) \mapsto (T, X, \theta', \varphi') 
\end{cases} 
\end{equation*}
with 
\begin{equation*} 
T = \frac{U + V}{2} \, , \quad X = \frac{- U + V}{2} \, , \quad \theta' = \theta \, , \quad \textnormal{and} \quad \varphi' = \varphi \quad \textnormal{for} \,\,\, \textnormal{B}_{\textnormal{I}} \cup \textnormal{B}_{\textnormal{II}} \, ,
\end{equation*}
where $T \in \bigl(|X - \pi/4| - \pi/4, - |X - \pi/4| + \pi/4\bigr)$ and $X \in (0, \pi/2)$ for $\textnormal{B}_{\textnormal{I}}$ and $T \in \bigl(|X|, \pi/4\bigr)$ and $X \in (- \pi/4, \pi/4)$ for $\textnormal{B}_{\textnormal{II}}$. Using these coordinates, the Schwarzschild metric can be represented as
\begin{equation*} 
\boldsymbol{g} = \frac{32 M^3 \, e^{- r/(2 M)}}{\bigl[\cos^2(T) - \sin^2(X)\bigr]^2 \, r} \, (\textnormal{d}T \otimes \textnormal{d}T - \textnormal{d}X \otimes \textnormal{d}X) - r^2 \, \boldsymbol{g}_{S^2} \, .
\end{equation*}
We note in passing that the compactified Kruskal--Szekeres time coordinate $T$ is a temporal function on $\textnormal{B}_{\textnormal{I}} \cup \textnormal{B}_{\textnormal{II}}$ (cf.\ Definition \ref{DefTTF} in Section \ref{IC}). Furthermore, even though compactified Kruskal--Szekeres spacetime coordinates are more general than required, they are---and yield representations of geometric quantities that are---nevertheless still fairly simple and easy to handle. However, if desired, one may as well work with different types of compactified horizon-penetrating coordinate systems derived from, e.g., Gullstrand--Painlev\'e coordinates, Lema\^{i}tre coordinates, or advanced Eddington--Finkelstein coordinates \cite{Gull, Pain, Lem}.

\subsection{Penrose Diagram of the Schwarzschild Black Hole Geometry}

\noindent Since the Schwarzschild black hole geometry $(\mathfrak{M}, \boldsymbol{g})$ consists of a product space $\mathfrak{M} = \mathfrak{M}^{(2)}_{\textnormal{L}} \times \mathfrak{M}^{(2)}_{\textnormal{R}}$, where the $2$-dimensional Lorentzian component is isomorphic to $\mathfrak{M}^{(2)}_{\textnormal{L}} \cong \mathfrak{M} \slash \textnormal{SO}(3) \cong \mathbb{R}^2$ and the $2$-dimensional Riemannian component to $\mathfrak{M}^{(2)}_{\textnormal{R}} \cong S^2$, and is furthermore endowed with a metric of the form
\begin{equation} \label{prodstr}
\boldsymbol{g} = \boldsymbol{g}^{(2)}_{\textnormal{L}} \oplus \bigl(f \hspace{0.02cm} \boldsymbol{g}^{(2)}_{\textnormal{R}}\bigr) \, ,
\end{equation}
in which $\boldsymbol{g}^{(2)}_{\textnormal{L}}$ and $\boldsymbol{g}^{(2)}_{\textnormal{R}}$ are $2$-dimensional Lorentzian and Riemannian metrics, respectively, and $f \colon \mathfrak{M}^{(2)}_{\textnormal{L}} \rightarrow \mathbb{R}_{> 0}$ is a smooth function (see, e.g., \cite{ONeill}), it allows for the natural identification 
\begin{equation*}
T\bigl(\mathfrak{M}^{(2)}_{\textnormal{L}} \times \mathfrak{M}^{(2)}_{\textnormal{R}}\bigr) = T\mathfrak{M}^{(2)}_{\textnormal{L}} \oplus T\mathfrak{M}^{(2)}_{\textnormal{R}} 
\end{equation*}
and, therefore, the splitting
\begin{equation*}
\boldsymbol{g}(\boldsymbol{Y}_{\textnormal{L}} + \boldsymbol{Y}_{\textnormal{R}}, \boldsymbol{Z}_{\textnormal{L}} + \boldsymbol{Z}_{\textnormal{R}}) = \boldsymbol{g}^{(2)}_{\textnormal{L}}(\boldsymbol{Y}_{\textnormal{L}}, \boldsymbol{Z}_{\textnormal{L}}) + f \hspace{0.02cm} \boldsymbol{g}^{(2)}_{\textnormal{R}}(\boldsymbol{Y}_{\textnormal{R}}, \boldsymbol{Z}_{\textnormal{R}}) 
\end{equation*}
for $\boldsymbol{Y}_k, \boldsymbol{Z}_k \in \Gamma\bigl(T\mathfrak{M}^{(2)}_k\bigr)$, $k \in \{\textnormal{L}, \textnormal{R}\}$. Accordingly, every null geodesic in the $2$-dimensional Lorentzian submanifold $(\mathfrak{M}^{(2)}_{\textnormal{L}}, \boldsymbol{g}^{(2)}_{\textnormal{L}})$ is also a null geodesic in the full $4$-dimensional Schwarzschild black hole geometry $(\mathfrak{M}, \boldsymbol{g})$ for a fixed Riemannian submanifold $(\mathfrak{M}^{(2)}_{\textnormal{R}}, \boldsymbol{g}^{(2)}_{\textnormal{R}})$, making the $2$-dimensional Lorentzian submanifold totally geodesic. Radial causal relations between different points in $(\mathfrak{M}, \boldsymbol{g})$ can thus be simply analyzed by means of $(\mathfrak{M}^{(2)}_{\textnormal{L}}, \boldsymbol{g}^{(2)}_{\textnormal{L}})$, in particular by using a Penrose diagram \cite{Penrose0, Penrose1, Carter2, Walker}, where the metric $\boldsymbol{g}^{(2)}_{\textnormal{L}}$ on this $2$-dimensional, finite-sized diagram is locally conformally equivalent to the $2$-dimensional Minkowski metric and every point of the diagram corresponds to a $2$-sphere. For the construction of this Penrose diagram, we employ the relations 
\begin{equation*} 
\left\{\!\begin{aligned}
& \, \displaystyle \frac{\sin{(2 T)}}{\sin{(2 X)}} = \tanh{\bigl(t/(4 M)\bigr)} \hspace{0.635cm} \textnormal{and} \quad \frac{\cos{(2 T)}}{\cos{(2 X)}} = - \coth{\bigl(r_{\star}/(4 M)\bigr)} \, \hspace{0.39cm} \textnormal{for} \,\,\, \textnormal{B}_{\textnormal{I}} \\[0.15cm] 
& \, \displaystyle \frac{\sin{(2 T)}}{\sin{(2 X)}} = - \coth{\bigl(t/(4 M)\bigr)} \quad \textnormal{and} \quad \frac{\cos{(2 T)}}{\cos{(2 X)}} = - \tanh{\bigl(r_{\star}/(4 M)\bigr)} \, \hspace{0.35cm} \textnormal{for} \,\,\, \textnormal{B}_{\textnormal{II}} 
\end{aligned}\right\} 
\end{equation*}
between the compactified Kruskal--Szekeres time and radial coordinates and the Schwarzschild time  and Regge--Wheeler coordinates, which directly lead to the asymptotics shown in TABLE I. These asymptotics may be used to define the relevant structures of the Penrose diagram, namely future/past timelike infinity $i^{\pm} = (T = \pm \pi/4, X = \pi/4)$, future/past null infinity $\mathscr{I}^{\pm} = \bigl\{(T, X) \, \big| \, T = \pm (- X + \pi/2) \,\,\, \textnormal{and} \,\,\, \pi/4 < X < \pi/2\bigr\}$, spacelike infinity $i^0 = (T = 0, X = \pi/2)$, the event horizon at $\bigl\{(T, X) \, \big| \, T = X \,\,\, \textnormal{and} \,\,\, 0 \leq X \leq \pi/4\bigr\}$, and the location of the curvature singularity at $\bigl\{(T, X) \, \big| \, T = \pi/4 \,\,\, \textnormal{and} \,\,\, - \pi/4 < X < \pi/4\bigr\}$. We depict the Penrose diagram of the Schwarzschild black hole geometry $\textnormal{B}_{\textnormal{I}} \cup \textnormal{B}_{\textnormal{II}}$ in FIG.~\ref{CPDSchwarzschild}.

\begin{table}[h] \label{table1}
\caption{Asymptotic relations between the Kruskal--Szekeres and Schwarzschild time and radial coordinates.}
\begin{ruledtabular}
\begin{tabular}{llll}
\\[-0.2cm]
& \hspace{1.9cm} $r \rightarrow \infty$ & \hspace{1.7cm} $r \rightarrow 2 M$ & \hspace{1.3cm} $r \rightarrow 0$ \\ \\
$\textnormal{B}_{\textnormal{I}}$ & $T = \pm \, [X - \pi/2], X \in [\pi/4, \pi/2]$ & \hspace{0.5cm} $T = \pm X, X \in [0, \pi/4]$ & \hspace{1.55cm} --- \\ \\
$\textnormal{B}_{\textnormal{II}}$ & \hspace{2.2cm} --- & \hspace{0.5cm} $T = |X|, X \in [- \pi/4, \pi/4]$ & $T = \pi/4, X \in (- \pi/4, \pi/4)$ \\[0.2cm]
\hline \\[-0.2cm]
& \hspace{1.9cm} $t \rightarrow \infty$ & \hspace{1.6cm} $t \rightarrow - \infty$ & \\ \\
$\textnormal{B}_{\textnormal{I}}$ & $T = - [X - \pi/2], X \in [\pi/4, \pi/2]$ & \hspace{0.5cm} $T = [X - \pi/2], X \in [\pi/4, \pi/2]$ & \\ \\
& $T = X, X \in [0, \pi/4]$ & \hspace{0.5cm} $T = - X, X \in [0, \pi/4]$ & \\ \\
$\textnormal{B}_{\textnormal{II}}$ & $T = - X, X \in [- \pi/4, 0]$ & \hspace{0.5cm} $T = X, X \in [0, \pi/4]$ & \\[0.2cm]
\end{tabular}
\end{ruledtabular}
\end{table}
%
%

\begin{figure}[h]%
\vspace{0.5cm}
\centering
\includegraphics[width=0.45\columnwidth]{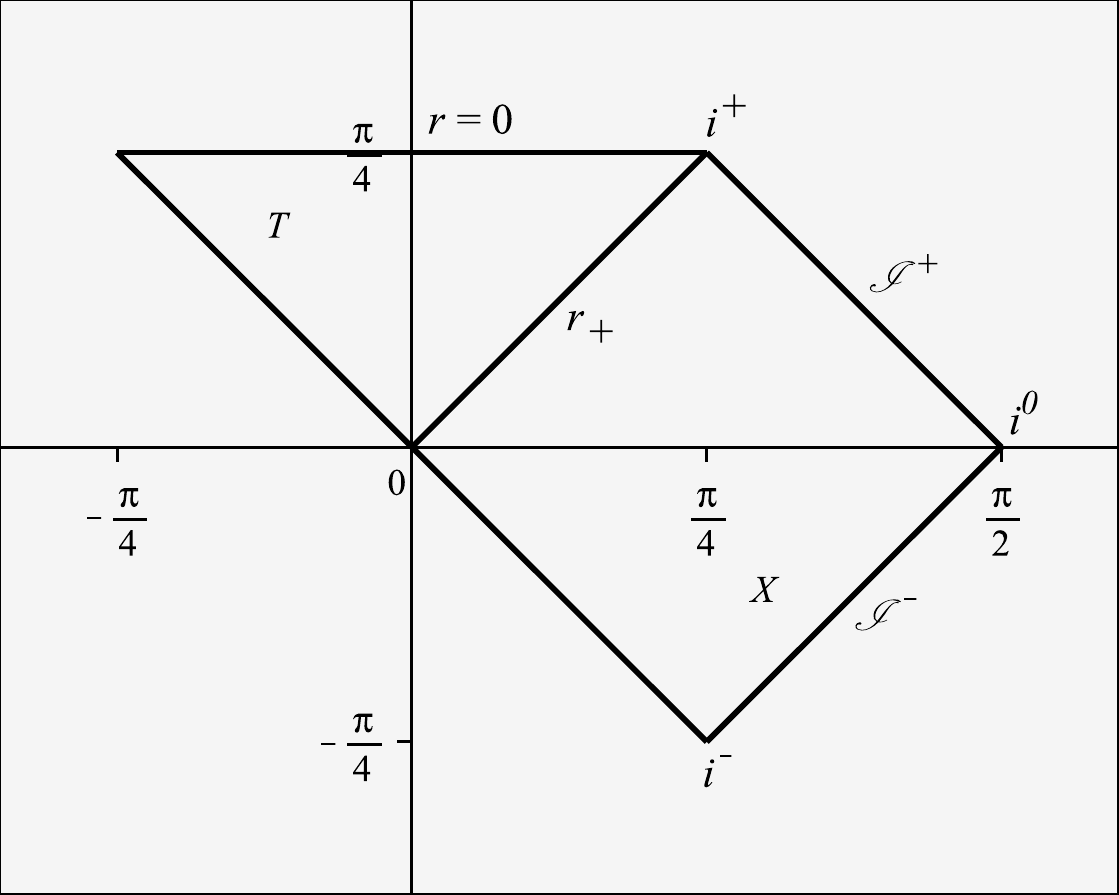}%
\caption[...]
{Penrose diagram of the exterior and interior regions of the Schwarzschild black hole geometry.}%
\label{CPDSchwarzschild}%
\end{figure}

\subsection{Cauchy Surfaces and Time-type Functions} \label{IC}

\noindent We now recall the concepts of Cauchy surfaces and time-type functions.
\begin{Def}
A \textit{Cauchy surface} of a connected, time-orientable Lorentzian manifold $(\mathfrak{M}, \boldsymbol{g})$ is any subset $\mathfrak{N} \subset \mathfrak{M}$ that is closed and achronal, and has the domain of dependence $D(\mathfrak{N}) = \mathfrak{M}$, i.e., it is intersected by every inextensible timelike curve exactly once. 
\end{Def}
\noindent A Cauchy surface is therefore a topological hypersurface \cite{ONeill}, which can be approximated by a smooth, spacelike hypersurface \cite{BernalSanchez2}. Moreover, if $(\mathfrak{M}, \boldsymbol{g})$ admits a Cauchy surface, it is globally hyperbolic \cite{RG2}. 
\begin{Def} \label{DefTTF}
We let $(\mathfrak{M}, \boldsymbol{g})$ be a connected, time-orientable Lorentzian manifold. A function $\mathfrak{t} \hspace{-0.07cm}: \mathfrak{M} \rightarrow \mathbb{R}$ is called a
\begin{itemize}
\item[1.] \textit{generalized time function} if it is strictly increasing on any future-directed causal curve.
\item[2.] \textit{time function} if it is a continuous generalized time function.
\item[3.] \textit{temporal function} if it is a smooth function with future-directed, timelike gradient $\boldsymbol{\nabla} \mathfrak{t} = g^{\mathfrak{t} \nu} \partial_{\nu}$.
\end{itemize}
\end{Def}
\noindent According to \cite{BernalSanchez1, RG}, there is the following relation between time-type functions and the notion of global hyperbolicity. 
\begin{Prp}
Any connected, time-orientable, globally hyperbolic Lorentzian manifold $(\mathfrak{M}, \boldsymbol{g})$ contains a \textit{Cauchy temporal function} $\mathfrak{t}$, that is, a temporal function for which the level sets $\mathfrak{t}^{- 1}(\, . \,)$ are smooth, spacelike Cauchy hypersurfaces $(\mathfrak{N}_{\mathfrak{t}})_{\mathfrak{t} \in \mathbb{R}}$ with $\mathfrak{N}_{\mathfrak{t}} := \{\mathfrak{t}\} \times \mathfrak{N}$ and $\mathfrak{N}_{\mathfrak{t}} \subset J^-(\mathfrak{N}_{\mathfrak{t}'})$ for all $\mathfrak{t} < \mathfrak{t}'$.
\end{Prp}
\noindent We remark that a coordinate system $(\mathfrak{t}, \boldsymbol{x})$ on $\mathfrak{M}$, where $\mathfrak{t} \in \mathbb{R}$ is a Cauchy temporal function and $\boldsymbol{x}$ are coordinates on $\mathfrak{N}$, may be understood as corresponding to an observer who is co-moving along the flow lines of the Killing field $\Gamma(T\mathfrak{M}) \ni \boldsymbol{K} = \partial_{\mathfrak{t}}$.

\section{Geometric Construction Procedure of Horizon-penetrating Cauchy Coordinates for the Schwarzschild Black Hole Geometry} \label{III}

\noindent We begin by simplifying the geometric shape of the Penrose diagram of the Schwarzschild black hole geometry $\textnormal{B}_{\textnormal{I}} \cup \textnormal{B}_{\textnormal{II}}$ transforming the trapezoid shown in FIG.\ \ref{Trafo0} into a centrally symmetric diamond as in FIG.\ \ref{Trafo5}. In more detail, we first rotate the trapezoid counter-clockwise about an angle of $\pi/4 \,\, \textnormal{rad}$ [FIG.\ \ref{Trafo0} $\rightarrow$ FIG.\ \ref{Trafo1}] employing the transformation
\begin{equation} \label{T1}
\mathfrak{T}^{(1)} \colon
\begin{cases}
\, \displaystyle \biggl(- \frac{\pi}{4}, \frac{\pi}{4}\biggr) \times \biggl(- \frac{\pi}{4}, \frac{\pi}{2}\biggr) \rightarrow \biggl(0, \frac{\pi}{2 \sqrt{2}}\biggr) \times \biggl(- \frac{\pi}{2 \sqrt{2}}, \frac{\pi}{2 \sqrt{2}}\biggr) \vspace{0.3cm} \\
\hspace{0.292cm} (T = T^{(0)}, X = X^{(0)}) \mapsto (T^{(1)}, X^{(1)}) 
\end{cases} 
\end{equation}
with 
\begin{equation*}
T^{(1)} = \frac{T^{(0)} + X^{(0)}}{\sqrt{2}} \quad \textnormal{and} \quad X^{(1)} = \frac{- T^{(0)} + X^{(0)}}{\sqrt{2}} \, ,
\end{equation*}
where $T^{(1)} < X^{(1)} + \pi/(2 \sqrt{2} \hspace{0.03cm} )$ for $- \pi/(2 \sqrt{2} \hspace{0.03cm} ) < X^{(1)} \leq 0$. We then deform the resulting trapezoid into a rectangle [FIG.\ \ref{Trafo1} $\rightarrow$ FIG.\ \ref{Trafo2}] by identifying the line 
\begin{equation*}
\bigl\{(T^{(1)}, X^{(1)}) \, \big| \, T^{(1)} = X^{(1)} + \pi/(2 \sqrt{2} \hspace{0.03cm} ) \,\,\,\, \textnormal{and} \,\,\, - \pi/(2 \sqrt{2} \hspace{0.03cm} ) \leq X^{(1)} \leq 0\bigr\}
\end{equation*}
with the line 
\begin{equation*}
\bigl\{(T^{(1)}, X^{(1)}) \, \big| \, 0 \leq T^{(1)} \leq \pi/(2 \sqrt{2} \hspace{0.03cm} ) \,\,\,\, \textnormal{and} \,\,\,\, X^{(1)} = - \pi/(2 \sqrt{2} \hspace{0.03cm} )\bigr\}
\end{equation*}
applying the transformation
\begin{equation} \label{T2}
\mathfrak{T}^{(2)} \colon
\begin{cases}
\, \displaystyle \biggl(0, \frac{\pi}{2 \sqrt{2}}\biggr) \times \biggl(- \frac{\pi}{2 \sqrt{2}}, \frac{\pi}{2 \sqrt{2}}\biggr) \rightarrow \biggl(0, \frac{\pi}{2 \sqrt{2}}\biggr) \times \biggl(- \frac{\pi}{2 \sqrt{2}}, \frac{\pi}{2 \sqrt{2}}\biggr) \vspace{0.3cm} \\
\hspace{2.70cm} (T^{(1)}, X^{(1)}) \mapsto (T^{(2)}, X^{(2)}) 
\end{cases} 
\end{equation}
with
\begin{equation*}
T^{(2)} = T^{(1)} \quad \textnormal{and} \quad X^{(2)} = \frac{T^{(1)}/2 - X^{(1)}}{T^{(1)} \sqrt{2}/\pi - 1} \, .
\end{equation*}

\begin{figure}[t]%
\centering
\subfigure[][]{%
\label{Trafo0}%
\includegraphics[width=0.3\columnwidth]{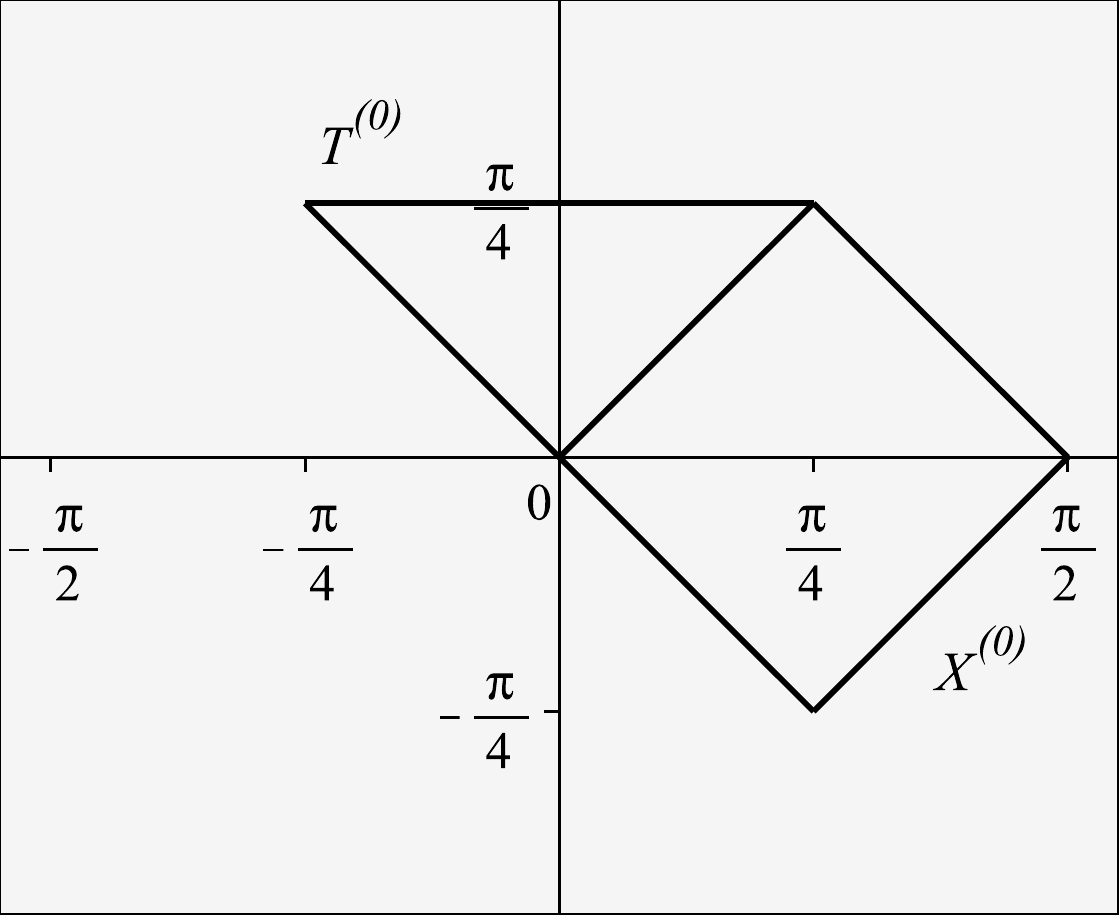}}%
\hspace{12pt}%
\subfigure[][]{%
\label{Trafo1}%
\includegraphics[width=0.3\columnwidth]{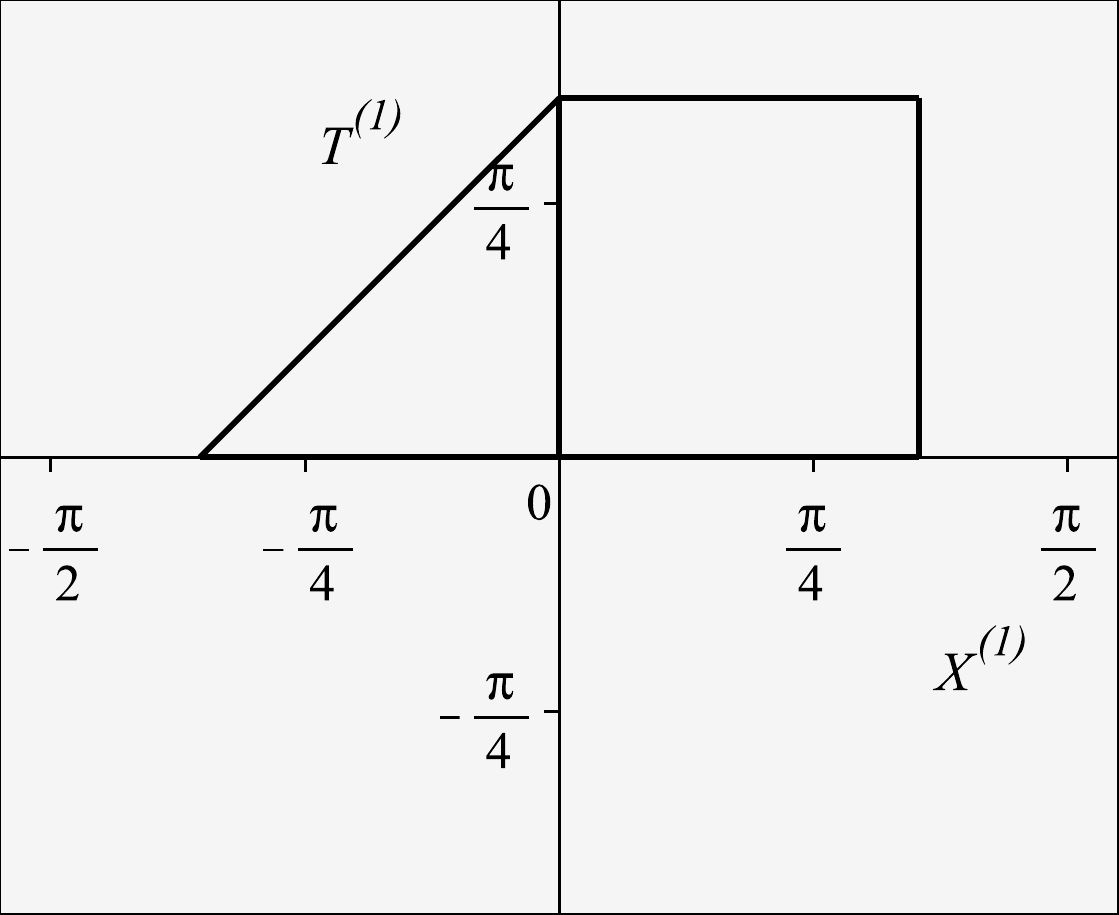}}%
\hspace{12pt}%
\subfigure[][]{%
\label{Trafo2}%
\includegraphics[width=0.3\columnwidth]{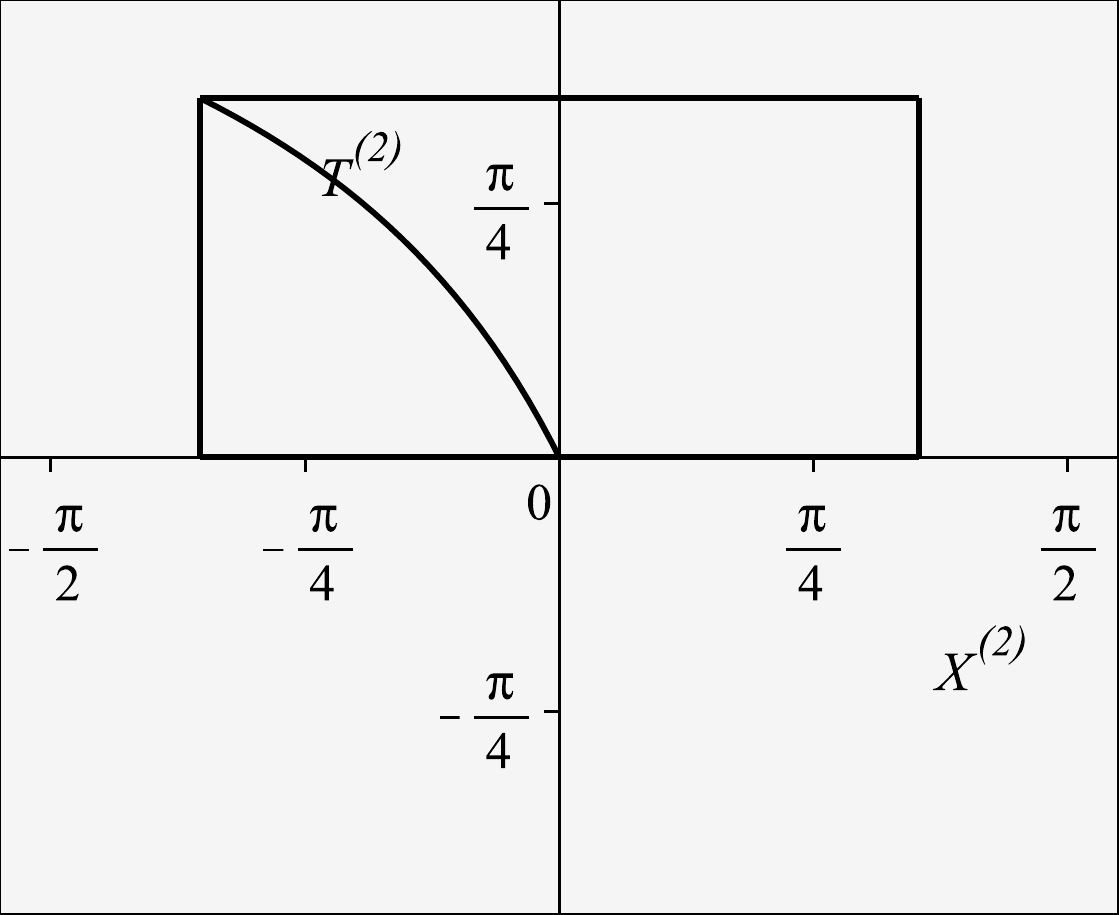}}\\[0.1cm]
\subfigure[][]{%
\label{Trafo3}%
\includegraphics[width=0.3\columnwidth]{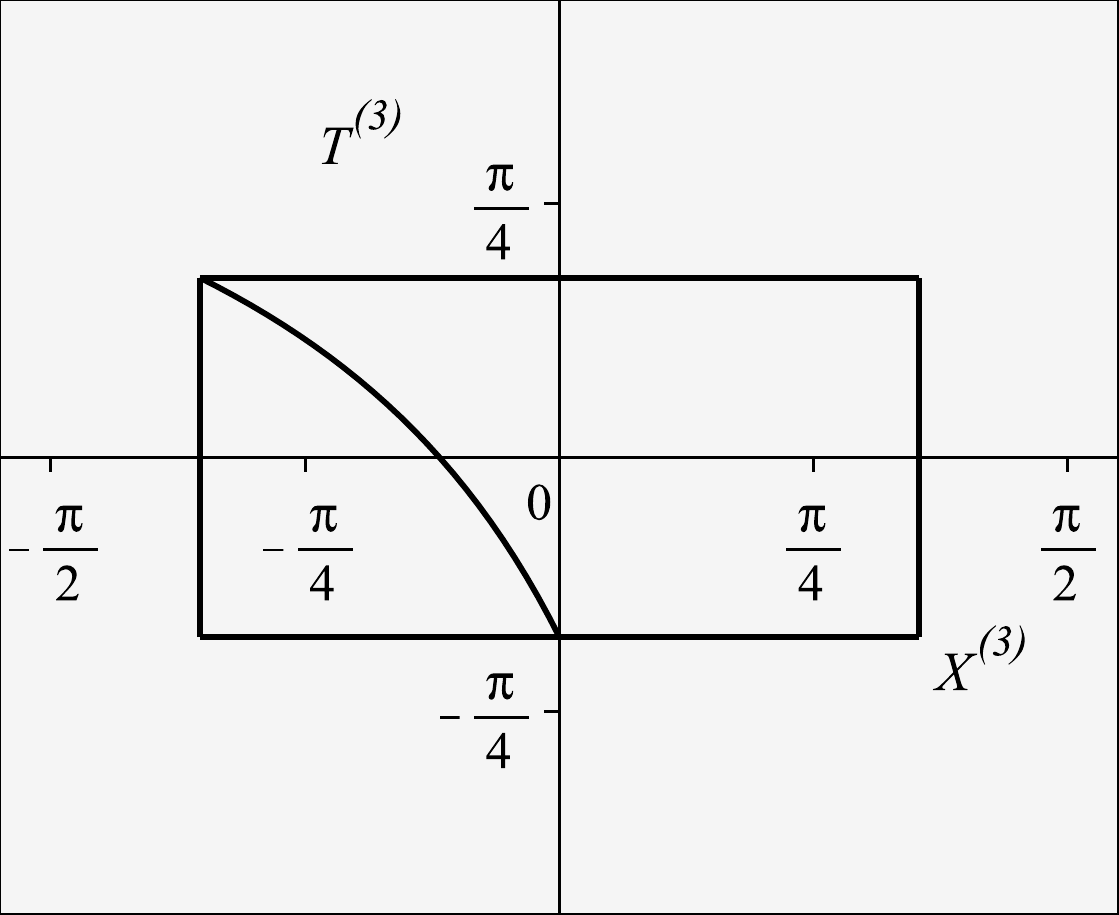}}%
\hspace{12pt}%
\subfigure[][]{%
\label{Trafo4}%
\includegraphics[width=0.3\columnwidth]{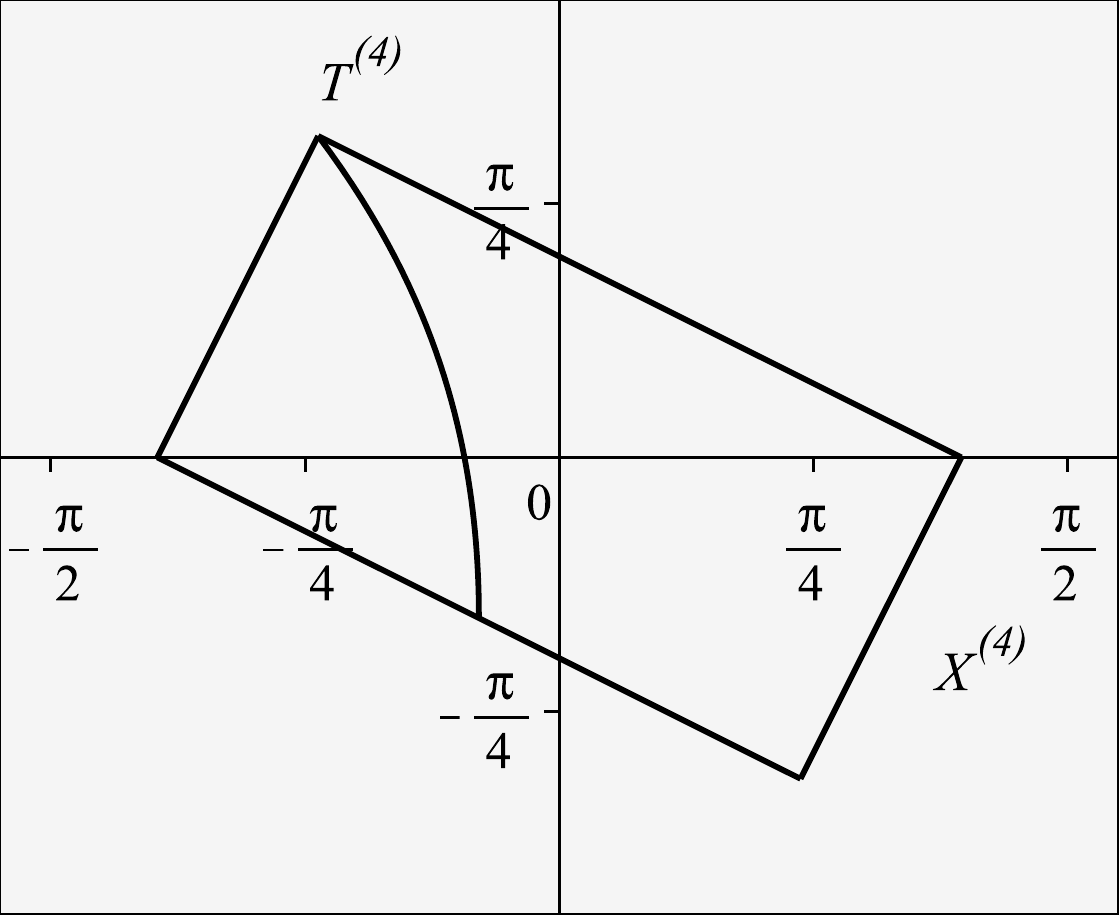}}%
\hspace{12pt}%
\subfigure[][]{%
\label{Trafo5}%
\includegraphics[width=0.3\columnwidth]{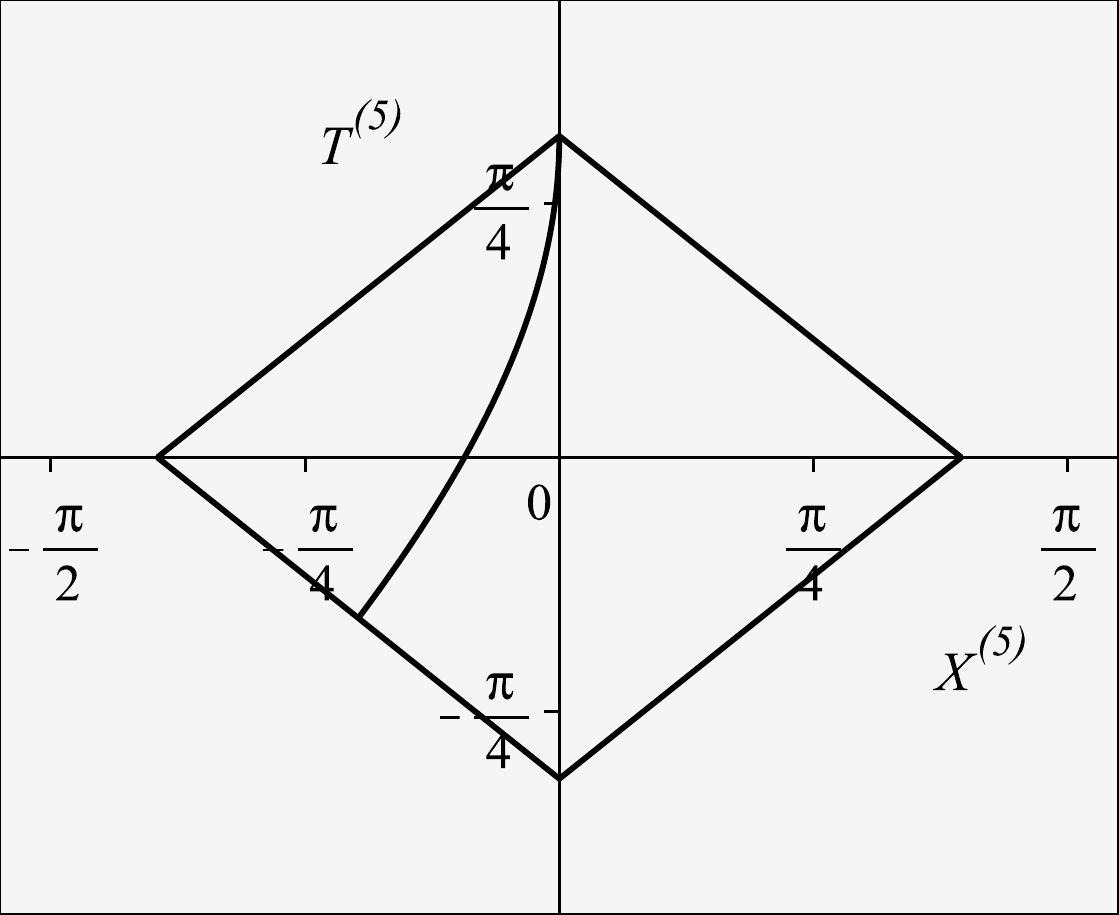}}%
\caption[...]
{Geometric representations of the transformations (\ref{T1})--(\ref{T5}).}%
\label{Trafo}%
\end{figure}

\noindent Subsequently, we translate the rectangle by the distance $- \pi/(4 \sqrt{2} \hspace{0.03cm} )$ along the ordinate [FIG.\ \ref{Trafo2} $\rightarrow$ FIG.\ \ref{Trafo3}] and rotate it clockwise about an angle of $\arctan{(1/2)} \,\, \textnormal{rad}$ [FIG.\ \ref{Trafo3} $\rightarrow$ FIG.\ \ref{Trafo4}] using the mappings
\begin{equation} \label{T3}
\mathfrak{T}^{(3)} \colon
\begin{cases}
\, \displaystyle \biggl(0, \frac{\pi}{2 \sqrt{2}}\biggr) \times \biggl(- \frac{\pi}{2 \sqrt{2}}, \frac{\pi}{2 \sqrt{2}}\biggr) \rightarrow \biggl(- \frac{\pi}{4 \sqrt{2}}, \frac{\pi}{4 \sqrt{2}}\biggr) \times \biggl(- \frac{\pi}{2 \sqrt{2}}, \frac{\pi}{2 \sqrt{2}}\biggr) \vspace{0.3cm} \\
\hspace{2.70cm} (T^{(2)}, X^{(2)}) \mapsto (T^{(3)}, X^{(3)})  
\end{cases} 
\end{equation}
with 
\begin{equation*}
T^{(3)} = T^{(2)} - \frac{\pi}{4 \sqrt{2}} \quad \textnormal{and} \quad X^{(3)} = X^{(2)}
\end{equation*}
and
\begin{equation} \label{T4}
\mathfrak{T}^{(4)} \colon
\begin{cases}
\, \displaystyle \biggl(- \frac{\pi}{4 \sqrt{2}}, \frac{\pi}{4 \sqrt{2}}\biggr) \times \biggl(- \frac{\pi}{2 \sqrt{2}}, \frac{\pi}{2 \sqrt{2}}\biggr) \rightarrow \biggl(- \frac{\pi}{\sqrt{10}}, \frac{\pi}{\sqrt{10}}\biggr) \times \biggl(- \sqrt{\frac{5}{2}} \, \frac{\pi}{4}, \sqrt{\frac{5}{2}} \, \frac{\pi}{4}\biggr) \vspace{0.3cm} \\
\hspace{3.534cm} (T^{(3)}, X^{(3)}) \mapsto (T^{(4)}, X^{(4)})  
\end{cases} 
\end{equation}
with 
\begin{equation*}
T^{(4)} = \frac{2 T^{(3)} - X^{(3)}}{\sqrt{5}} \quad \textnormal{and} \quad X^{(4)} = \frac{T^{(3)} + 2 X^{(3)}}{\sqrt{5}} \, ,
\end{equation*}
respectively, where 
\begin{equation*}
\begin{split}
T^{(4)} & < \Theta\biggl(X^{(4)} + \sqrt{\frac{5}{2}} \, \frac{\pi}{4}\biggr) \, \Theta\biggl(- \frac{3 \pi}{4 \sqrt{10}} - X^{(4)}\biggr) \biggl[2 X^{(4)} + \sqrt{\frac{5}{2}} \, \frac{\pi}{2}\biggr] \\ \\
& \hspace{0.4cm} + \Theta\biggl(X^{(4)} + \frac{3 \pi}{4 \sqrt{10}}\biggr) \, \Theta\biggl(\sqrt{\frac{5}{2}} \, \frac{\pi}{4} - X^{(4)}\biggr) \frac{1}{2} \biggl[- X^{(4)} + \sqrt{\frac{5}{2}} \, \frac{\pi}{4}\biggr] 
\end{split}
\end{equation*}
and 
\begin{equation*}
\begin{split}
& - \Theta\biggl(X^{(4)} + \sqrt{\frac{5}{2}} \, \frac{\pi}{4}\biggr) \, \Theta\biggl(\frac{3 \pi}{4 \sqrt{10}} - X^{(4)}\biggr) \frac{1}{2} \biggl[X^{(4)} + \sqrt{\frac{5}{2}} \, \frac{\pi}{4}\biggr] \\ \\
& + \Theta\biggl(X^{(4)} - \frac{3 \pi}{4 \sqrt{10}}\biggr) \, \Theta\biggl(\sqrt{\frac{5}{2}} \, \frac{\pi}{4} - X^{(4)}\biggr) \biggl[2 X^{(4)} - \sqrt{\frac{5}{2}} \, \frac{\pi}{2}\biggr] < T^{(4)} \, .
\end{split}
\end{equation*}
Here, $\Theta(\, . \,) := [1 + \textnormal{sgn}(\, . \,)]/2$ is the Heaviside step function. Lastly, we employ the shear transformation
\begin{equation} \label{T5}
\mathfrak{T}^{(5)} \colon
\begin{cases}
\, \displaystyle \biggl(- \frac{\pi}{\sqrt{10}}, \frac{\pi}{\sqrt{10}}\biggr) \times \biggl(- \sqrt{\frac{5}{2}} \, \frac{\pi}{4}, \sqrt{\frac{5}{2}} \, \frac{\pi}{4}\biggr) \rightarrow \biggl(- \frac{\pi}{\sqrt{10}}, \frac{\pi}{\sqrt{10}}\biggr) \times \biggl(- \sqrt{\frac{5}{2}} \, \frac{\pi}{4}, \sqrt{\frac{5}{2}} \, \frac{\pi}{4}\biggr) \vspace{0.3cm} \\
\hspace{4.01cm} (T^{(4)}, X^{(4)}) \mapsto (T^{(5)}, X^{(5)})  
\end{cases} 
\end{equation}
with 
\begin{equation*}
V = T^{(5)} = T^{(4)} \quad \textnormal{and} \quad W = X^{(5)} = \frac{3 T^{(4)}}{4} + X^{(4)} \, ,
\end{equation*}
where $4 \hspace{0.02cm} |X^{(5)}|/5 - \pi/\sqrt{10} < T^{(5)} < - 4 \hspace{0.02cm} |X^{(5)}|/5 + \pi/\sqrt{10}$, in order to obtain the centrally symmetric diamond [FIG.\ \ref{Trafo4} $\rightarrow$ FIG.\ \ref{Trafo5}]. The composition of the transformations (\ref{T1})--(\ref{T5}) yields the relations 

\begin{equation} \label{VWTX}
\begin{split}
V & = \frac{2}{\sqrt{10} \, [\pi - X - T]} \, \biggl(- (X + T)^2 + \pi \, \biggl[X + 2 T - \frac{\pi}{4}\biggr]\biggr) \\ \\
W & = \frac{5}{2 \sqrt{10} \, [\pi - X - T]} \, \biggl(- (X + T)^2 + \frac{\pi}{2} \, \biggl[3 X + T - \frac{\pi}{2}\biggr]\biggr) \, . 
\end{split}
\end{equation}

Next, we formulate conditions for the determination of specific indexed families of smooth functions $\bigl(V_{\lambda}(W) \, \big| \, \lambda \in \mathbb{R}\bigr)$ that foliate the diamond: 
\vspace{0.1cm}
\begin{itemize} 
\item[$(\mathcal{C}1)$] Limit conditions: $\displaystyle V_{\pm \infty}(W) = \pm \frac{4}{5} \bigl[- |W| + \mu\bigr]$ \vspace{0.18cm}
\item[$(\mathcal{C}2)$] Boundary conditions: 
$V_{\lambda}(\pm \mu) = 0 \quad \forall \, \lambda \in \mathbb{R}$ \vspace{0.25cm}
\item[$(\mathcal{C}3)$] Smoothness condition: $V_{|\lambda| < \infty}(W) \in C^{\infty}\bigl((- \mu, \mu), \mathbb{R}\bigr)$ 
\item[$(\mathcal{C}4)$] Causality conditions: $\displaystyle - \frac{4}{5} < \partial_W V_{\lambda} < \frac{8}{5} \, \frac{4 W - \pi \sqrt{10}}{10 V_{\lambda} - \pi \sqrt{10}} \quad \textnormal{and} \quad 0 < \partial_{\lambda} V_{\lambda} \quad \forall \, W \in (- \mu, \mu)$ \vspace{0.08cm}
\item[$(\mathcal{C}5)$] Symmetry condition: $\lambda \mapsto - \lambda \,\, \Leftrightarrow \,\, (V_{\lambda}, W) \mapsto (- V_{\lambda}, W)$ ,
\end{itemize} \vspace{0.1cm}
where $\mu := \sqrt{5/2} \, \pi/4$. We note that the limit conditions in $(\mathcal{C}1)$ define the geometrical shape of the diamond, while the boundary conditions in $(\mathcal{C}2)$ specify the starting point and the endpoint of the functions $V_{\lambda}$. Besides, the first boundary condition $(V_{\lambda}, W) = (0, + \mu)$ gives rise to asymptotic flatness at spacelike infinity, whereas the second boundary condition $(V_{\lambda}, W) = (0, - \mu)$ ensures that the functions hit the curvature singularity only asymptotically. The meaning of the smoothness condition in $(\mathcal{C}3)$ is obvious. Moreover, the causality conditions in $(\mathcal{C}4)$ constrain the functions to be spacelike on the one hand, and nonintersecting on the other. Direct computations show that these conditions imply that the gradient $\boldsymbol{\nabla} \lambda$ on $\textnormal{B}_{\textnormal{I}} \cup \textnormal{B}_{\textnormal{II}}$ is future-directed and timelike, and hence that $\lambda$ is a temporal function. The reflection symmetry provided by the symmetry condition in $(\mathcal{C}5)$, however, is only incorporated for convenience. Therefore, it is not strictly required and may be dropped if desired. 

Finally, we regard the indices $\lambda$ of these families as time variables of global coordinate systems on $\textnormal{B}_{\textnormal{I}} \cup \textnormal{B}_{\textnormal{II}}$ determined by the general transformation 
\begin{equation*} 
\mathfrak{T}^{\textnormal{C}} \colon
\begin{cases}
\, \displaystyle \biggl(- \frac{\pi}{4}, \frac{\pi}{4}\biggr) \times \biggl(- \frac{\pi}{4}, \frac{\pi}{2}\biggr) \times (0, \pi) \times [0, 2 \pi) \rightarrow \mathbb{R} \times \biggl(- \frac{\pi}{4}, \frac{\pi}{2}\biggr) \times (0, \pi) \times [0, 2 \pi) \vspace{0.3cm} \\
\hspace{4.51cm} (T, X, \theta, \varphi) \mapsto (\lambda, X', \theta', \varphi')
\end{cases} 
\end{equation*}
with 
\begin{equation*}
\lambda = \lambda(T, X) \, , \quad X' = X \, , \quad \theta' = \theta \, , \quad \textnormal{and} \quad \varphi' = \varphi \, .
\end{equation*}	
We now prove that the level sets of the time variables $\lambda$, which are by the above construction smooth, spacelike, nonintersecting, asymptotically flat, and foliate the entire Schwarzschild black hole geometry, constitute Cauchy hypersurfaces.
\begin{Prp} 
We let $\mathfrak{S} \equiv \mathfrak{S}_{\lambda_0}$ be homeomorphic to the subset 
\begin{equation*}
\{\lambda_0\} \times \biggl(- \frac{\pi}{4}, \frac{\pi}{2}\biggr) \times S^2 \subset \textnormal{B}_{\textnormal{I}} \cup \textnormal{B}_{\textnormal{II}}
\end{equation*}
of the joined exterior and interior regions of the Schwarzschild black hole geometry, where this subset is a level set of the time coordinates $\lambda$ at $\lambda_0 = \textnormal{const.}$ Then, $\mathfrak{S}$ is a Cauchy hypersurface.
\end{Prp}
\begin{proof}	
We begin by noting that $\mathfrak{S}$ is closed in $\textnormal{B}_{\textnormal{I}} \cup \textnormal{B}_{\textnormal{II}}$, which is an immediate consequence of the fact that its complement 
\begin{equation*}
\displaystyle \mathfrak{S}^{\textnormal{c}} \cong \mathbb{R} \backslash \{\lambda_0\} \times \biggl(- \frac{\pi}{4}, \frac{\pi}{2}\biggr) \times S^2
\end{equation*}
is open. Moreover, as the time coordinates $\lambda$ are temporal functions on $\textnormal{B}_{\textnormal{I}} \cup \textnormal{B}_{\textnormal{II}}$, that is, $\textnormal{B}_{\textnormal{I}} \cup \textnormal{B}_{\textnormal{II}}$ is stably causal \cite{MingSan}, any connected causal curve through this region can intersect $\mathfrak{S}$ at most once. Thus, $\mathfrak{S}$ is achronal. It remains to be shown that the domain of dependence $D(\mathfrak{S}) = \textnormal{B}_{\textnormal{I}} \cup \textnormal{B}_{\textnormal{II}}$. To this end, it suffices to demonstrate that the total Cauchy horizon $H(\mathfrak{S})$ of $\mathfrak{S}$ is empty using a proof by contradiction. Hence, we suppose that there exists a point $p$ in the future Cauchy horizon $H^+(\mathfrak{S})$. Since $\mathfrak{S}$ is achronal and edgeless, $p$ is the future endpoint of a null geodesic $\gamma \subset H^+(\mathfrak{S})$, which is past inextensible in $\textnormal{B}_{\textnormal{I}} \cup \textnormal{B}_{\textnormal{II}}$ \cite{Wald}. From this, it follows that $\gamma \subset J^+(\mathfrak{S}) \cap J^-(p)$. Furthermore, as $\textnormal{B}_{\textnormal{I}} \cup \textnormal{B}_{\textnormal{II}}$ is globally hyperbolic, $J^+(\mathfrak{S}) \cap J^-(p)$ is contained in a compact set. And given that $\gamma$ cannot be imprisoned in a compact set that is stably causal \cite{Min}, we are led to a contradiction. Accordingly, $H^+(\mathfrak{S}) = \emptyset$. Due to time duality, we can argue that the same holds true for $H^-(\mathfrak{S})$, and therefore $H(\mathfrak{S}) = \emptyset$. 

\end{proof}

\section{Application to an Integrated Algebraic Sigmoid Function} \label{IV}

\begin{figure}[t]%
\centering
\subfigure[][]{%
\label{SFCurves}%
\includegraphics[width=0.45\columnwidth]{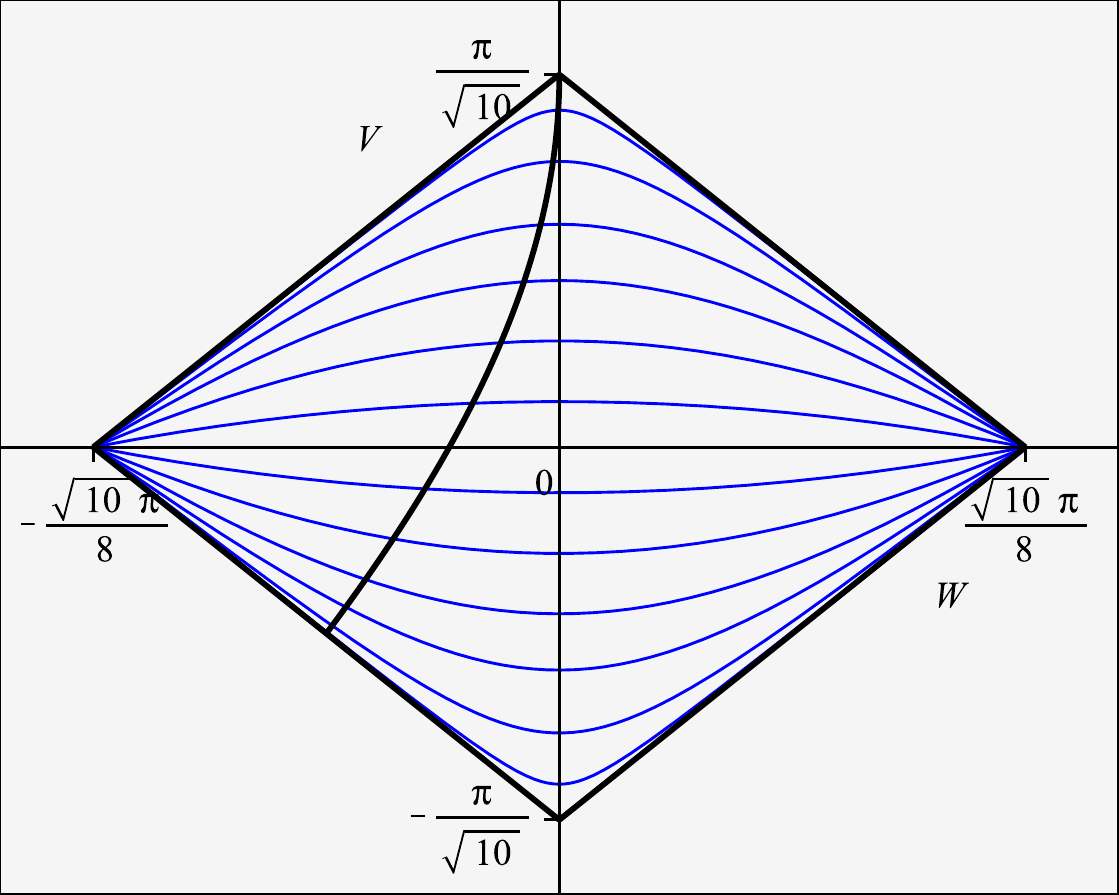}}
\hspace{15pt}%
\subfigure[][]{%
\label{CPDSchwarzschildF}%
\includegraphics[width=0.45\columnwidth]{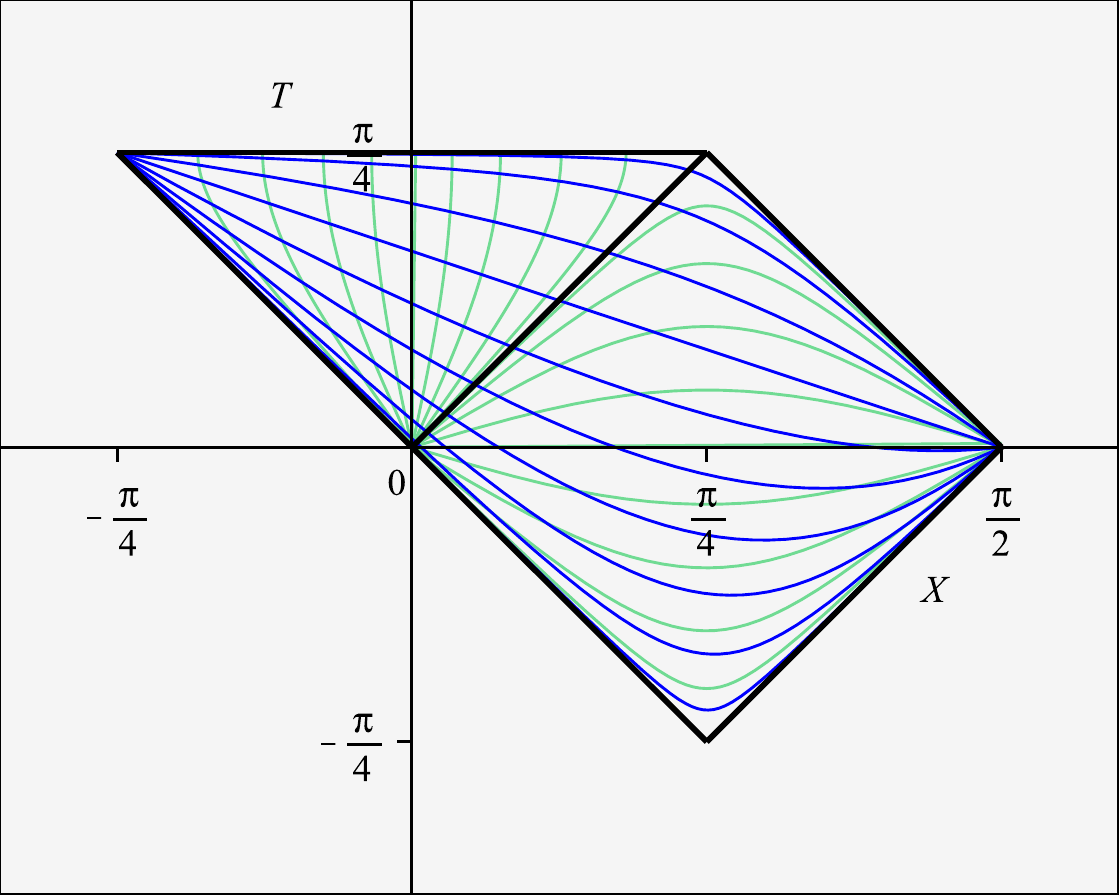}}%
\caption[...]
{Diamond representation of the Schwarzschild black hole geometry with smooth functions $V_{\lambda}(W)$ defined in Equation (\ref{fssfc}) for index values $\lambda \in \pm \{0, 0.2, 0.5, 0.9, 1.5, 3, 8\}$ \subref{SFCurves} and Penrose diagram of the Schwarzschild black hole geometry with level sets of the Cauchy temporal function $\lambda$ specified in Equation (\ref{CauchyTempFunc}) for values in $\{- 10, - 3.2, - 1.6, - 0.9, - 0.45, - 0.1, 0.28, 0.8, 2.1, 6.5\}$ (blue curves) and with level sets of the normalized Schwarzschild time coordinate $t/M \in \pm \{0, 1.24, 2.77, 4.75, 7.78\}$ for $\textnormal{B}_{\textnormal{I}}$ and $t/M \in \pm \{0, 0.86, 1.96, 3.58, 6.09\}$ for $\textnormal{B}_{\textnormal{II}}$ (aquamarine curves) for comparison \subref{CPDSchwarzschildF}.}%
\label{SchwarzschildFull}%
\end{figure}

\noindent In this section, we study an example of the families $\bigl(V_{\lambda}(W) \, \big| \, \lambda \in \mathbb{R}\bigr)$ that is based on an integrated algebraic sigmoid function. To be more precise, since our $2$-dimensional diagrammatic representation of the Schwarzschild black hole geometry $\textnormal{B}_{\textnormal{I}} \cup \textnormal{B}_{\textnormal{II}}$ is in the form of a centrally symmetric diamond, we are essentially interested in a smooth approximation of the absolute value function $|W|$ [see the limit conditions in $(\mathcal{C}1)$]. By considering the derivative of the absolute value function, namely the signum function $\textnormal{sgn}(W)$, we may easily find such a smooth approximation in terms of the integral of a hyperbolic tangent, an arctangent function, or an algebraic function. In the following, we work out the horizon-penetrating Cauchy coordinate system and the corresponding metric representation associated with the algebraic sigmoid function approximation
\begin{equation} \label{sgn}
\textnormal{sgn}(W) \approx \frac{|\lambda| W}{\sqrt{1 + \lambda^2 W^2}} \, , 
\end{equation}
where $\lambda$ serves as approximation parameter, because this simple example can be treated completely analytically. Thus, integrating Equation (\ref{sgn}) and imposing Conditions $(\mathcal{C}1)$--$(\mathcal{C}5)$, we obtain 
\begin{equation} \label{fssfc}
V_{\lambda}(W) = \frac{4}{5 \lambda} \Bigl[\sqrt{1 + \lambda^2 \mu^2} - \sqrt{1 + \lambda^2 W^2} \, \Bigr]  
\end{equation}
[for an illustration, see FIG.\ \ref{SFCurves}]. Inverting this expression with respect to $\lambda$ and substituting the relations specified in Equation (\ref{VWTX}) gives rise to the transformation from compactified Kruskal--Szekeres spacetime coordinates into horizon-penetrating Cauchy coordinates 
\begin{equation*} 
\mathfrak{T}^{\textnormal{C}} \colon
\begin{cases}
\, \displaystyle \biggl(- \frac{\pi}{4}, \frac{\pi}{4}\biggr) \times \biggl(- \frac{\pi}{4}, \frac{\pi}{2}\biggr) \times (0, \pi) \times [0, 2 \pi) \rightarrow \mathbb{R} \times \biggl(- \frac{\pi}{4}, \frac{\pi}{2}\biggr) \times (0, \pi) \times [0, 2 \pi) \vspace{0.3cm} \\
\hspace{4.51cm} (T, X, \theta, \varphi) \mapsto (\lambda, X', \theta', \varphi')
\end{cases} 
\end{equation*}
with 
\begin{equation} \label{CauchyTempFunc}
\lambda = \frac{5 V(T, X)}{\displaystyle 2 \sqrt{\biggl[W(T, X)^2 - \frac{25 V(T, X)^2}{16} + \frac{5 \pi^2}{32}\biggr]^2 - \frac{5 \pi^2 W(T, X)^2}{8}}} \, , \quad X' = X \, ,\quad \theta' = \theta \, , \quad \textnormal{and} \quad \varphi' = \varphi \, .
\end{equation}
The Schwarzschild metric formulated in these coordinates reads
\begin{equation*}
\begin{split}
\boldsymbol{g} & = \frac{32 M^3 \, e^{- r/(2 M)}}{\bigl[15 \pi \lambda + \sqrt{160 + 25 \pi^2 \lambda^2} \, \bigr]^4 \, \bigl[\cos^2(T) - \sin^2(X)\bigr]^2 \, r} \, \Bigl[400 \pi^2 \, \mathscr{C}_{\lambda}^2 \, \textnormal{d}\lambda \otimes \textnormal{d}\lambda - 40 \pi \, \mathscr{C}_{\lambda} \mathscr{C}_{X'} \\ \\
& \hspace{0.5cm} \times (\textnormal{d}\lambda \otimes \textnormal{d}X' + \textnormal{d}X' \otimes \textnormal{d}\lambda) + \bigl(4 \, \mathscr{C}_{X'}^2 - \bigl[15 \pi \lambda + \sqrt{160 + 25 \pi^2 \lambda^2} \, \bigr]^4 \hspace{0.04cm} \bigr) \, \textnormal{d}X' \otimes \textnormal{d}X' \, \Bigr] - r^2 \, \boldsymbol{g}_{S^2} \, ,
\end{split}
\end{equation*}
where
\begin{equation*}
\begin{split}
& \mathscr{C}_{\lambda} := \frac{20 \, [4 X' -5 \pi]}{\sqrt{160 + 25 \pi^2 \lambda^2}} - \frac{\sqrt{5 \pi}}{\sqrt{5 \pi \lambda^2 \, [X' - \pi/4]^2 - 2 \, [4 X' - 3 \pi]}} \biggl[\frac{10 \lambda \, (8 X'^2 - 2 \pi X' - \pi^2)}{\sqrt{160 + 25 \pi^2 \lambda^2}} + 3 \, (4 X' - 3 \pi)\biggr] \\ \\
& \mathscr{C}_{X'} := 10 \, \bigl[5 \pi^2 \lambda^2 + \pi \lambda \sqrt{160 + 25 \pi^2 \lambda^2} + 8\bigr] + \frac{\sqrt{5 \pi} \, \bigl(15 \pi \lambda + \sqrt{160 + 25 \pi^2 \lambda^2} \, \bigr) \bigl(5 \pi \lambda^2 \, [X' - \pi/4] - 4\bigr)}{\sqrt{5 \pi \lambda^2 \, [X' - \pi/4]^2 - 2 \, [4 X' - 3 \pi]}} \, .
\end{split}
\end{equation*}
We depict the foliation of the Schwarzschild black hole geometry by the level sets of $\lambda$ in the Penrose diagram in FIG.\ \ref{CPDSchwarzschildF}.

\section{Generalization to the Reissner--Nordstr\"om Black Hole Geometry} \label{V}

\noindent We generalize our results to the Reissner--Nordstr\"om black hole geometry up to the Cauchy horizon. This spacetime is, like the Schwarzschild black hole geometry, a connected, smooth, globally hyperbolic and asymptotically flat Lorentzian $4$-manifold ($\mathfrak{M}, \boldsymbol{g}$) with $\mathfrak{M}$ being homeomorphic to $\mathbb{R}^2 \times S^2$. It is, however, based on the $2$-parameter family of exact, spherically symmetric solutions $\boldsymbol{g}$ of the more general Einstein--Maxwell equations, which can be used to account for the final equilibrium state in the time evolution of the gravitational field of an isolated, spherically symmetric, electrically charged black hole. We begin by performing the replacement
\begin{equation*}
1 - \frac{2 M}{r} \,\, \rightarrow \,\, 1 - \frac{2 M}{r} + \frac{Q^2}{r^2} =: \frac{\Delta(r)}{r^2}
\end{equation*}
in the $g_{t t}$ and $g_{r r}$ components of the Schwarzschild metric (\ref{BLmetric}), where the parameter $Q \in \mathbb{R}$ denotes the electrical charge of the black hole geometry satisfying the relation $0 < |Q| < M$, and the two real-valued roots $r_{\pm} := M \pm \sqrt{M^2 - Q^2}$ of the function $\Delta \hspace{-0.07cm}: \mathbb{R}_{> 0} \rightarrow [- M^2 + Q^2, \infty)$ define an outer and an inner event horizon, respectively. This replacement gives rise to the Schwarzschild-type representation of the nonextreme Reissner--Nordstr\"om metric \cite{Reiss, Nord}
\begin{equation} \label{RNM}
\boldsymbol{g} = \frac{\Delta}{r^2} \, \textnormal{d}t \otimes \textnormal{d}t - \frac{r^2}{\Delta} \, \textnormal{d}r \otimes \textnormal{d}r - r^2 \, \boldsymbol{g}_{S^2} \, .
\end{equation}
We point out that the canonical Reissner--Nordstr\"om black hole geometry comprises the three connected components $\textnormal{B}_{\textnormal{I}} := \mathbb{R} \times \mathbb{R}_{> r_+} \times S^2$, $\textnormal{B}_{\textnormal{II}} := \mathbb{R} \times (r_-, r_+) \times S^2$, and $\textnormal{B}_{\textnormal{III}} := \mathbb{R} \times (0, r_-) \times S^2$, which have a causal structure that is qualitatively different from the one of the Schwarzschild case as $\textnormal{B}_{\textnormal{III}}$ contains a curvature singularity at $r = 0$ with timelike character and, more importantly for the present purpose, the inner event horizon at $r = r_-$ is a Cauchy horizon. Consequently, since our geometric construction procedure requires the underlying Lorentzian $4$-manifold to be globally hyperbolic, we consider only the region $\textnormal{B}_{\textnormal{I}} \cup \textnormal{B}_{\textnormal{II}}$ of the Reissner--Nordstr\"om black hole geometry up to the Cauchy horizon. We then transform the Schwarzschild-type coordinates into compactified Kruskal--Szekeres-type spacetime coordinates
\begin{equation*} 
\mathfrak{T}^{\textnormal{KS}} \colon
\begin{cases}
\, \displaystyle \mathbb{R} \times \mathbb{R}_{> 0} \times (0, \pi) \times [0, 2 \pi) \rightarrow \biggl(- \frac{\pi}{4}, \frac{\pi}{2}\biggr) \times \biggl(- \frac{\pi}{4}, \frac{\pi}{2}\biggr) \times (0, \pi) \times [0, 2 \pi) \vspace{0.3cm} \\
\hspace{2.57cm} (t, r, \theta, \varphi) \mapsto (T, X, \theta', \varphi')
\end{cases} 
\end{equation*}
with 
\begin{equation*} 
\left\{\!\begin{aligned}
& \, T = \displaystyle \frac{1}{2} \arctan{\biggl(\frac{\sinh{(\alpha t)}}{\cosh{(\alpha r_{\star})}}\biggr)} & & \textnormal{and} \quad X = \displaystyle - \frac{1}{2} \arctan{\biggl(\frac{\cosh{(\alpha t)}}{\sinh{(\alpha r_{\star})}}\biggr)} + \frac{\pi \Theta(r_{\star})}{2} & & \, \textnormal{for} \,\,\, \textnormal{B}_{\textnormal{I}} \\[0.1cm] 
& \, T = \displaystyle - \frac{1}{2} \arctan{\biggl(\frac{\cosh{(\alpha t)}}{\sinh{(\alpha r_{\star})}}\biggr)} & & \textnormal{and} \quad X = \displaystyle - \frac{1}{2} \arctan{\biggl(\frac{\sinh{(\alpha t)}}{\cosh{(\alpha r_{\star})}}\biggr)} & & \, \textnormal{for} \,\,\, \textnormal{B}_{\textnormal{II}} 
\end{aligned}\right\} \, , 
\end{equation*}
$\theta' = \theta$, and $\varphi' = \varphi$, where $T \in \bigl(|X - \pi/4| - \pi/4, - |X - \pi/4| + \pi/4\bigr)$ and $X \in (0, \pi/2)$ for $\textnormal{B}_{\textnormal{I}}$ and $T \in \bigl(|X|, \pi/2 - |X|\bigr)$ and $X \in (- \pi/4, \pi/4)$ for $\textnormal{B}_{\textnormal{II}}$. Here, the Regge--Wheeler coordinate is defined as
\begin{equation*}
r_{\star} := r + \frac{r_+^2}{r_+ - r_-} \, \ln{\bigg|\frac{r}{r_+} - 1\bigg|} - \frac{r_-^2}{r_+ - r_-} \, \ln{\bigg(\frac{r}{r_-} - 1\bigg)} 
\end{equation*}
and $\alpha := (r_+ - r_-)/(2 r_+^2)$ is a positive constant. The Reissner--Nordstr\"om metric (\ref{RNM}) written in terms of these coordinates takes the form
\begin{equation*}
\boldsymbol{g} = \frac{r_+ \, r_- \, (r/r_- - 1)^{1 - r^2_-/r^2_+} \, e^{- 2 \alpha r}}{\alpha^2 \bigl[\cos^2(T) - \sin^2(X)\bigr]^2 \, r^2} \, (\textnormal{d}T \otimes \textnormal{d}T - \textnormal{d}X \otimes \textnormal{d}X) - r^2 \, \boldsymbol{g}_{S^2} \, .
\end{equation*}

\begin{figure}[t]%
\centering
\subfigure[][]{%
\label{CPDReissnerNordstroem}%
\includegraphics[width=0.45\columnwidth]{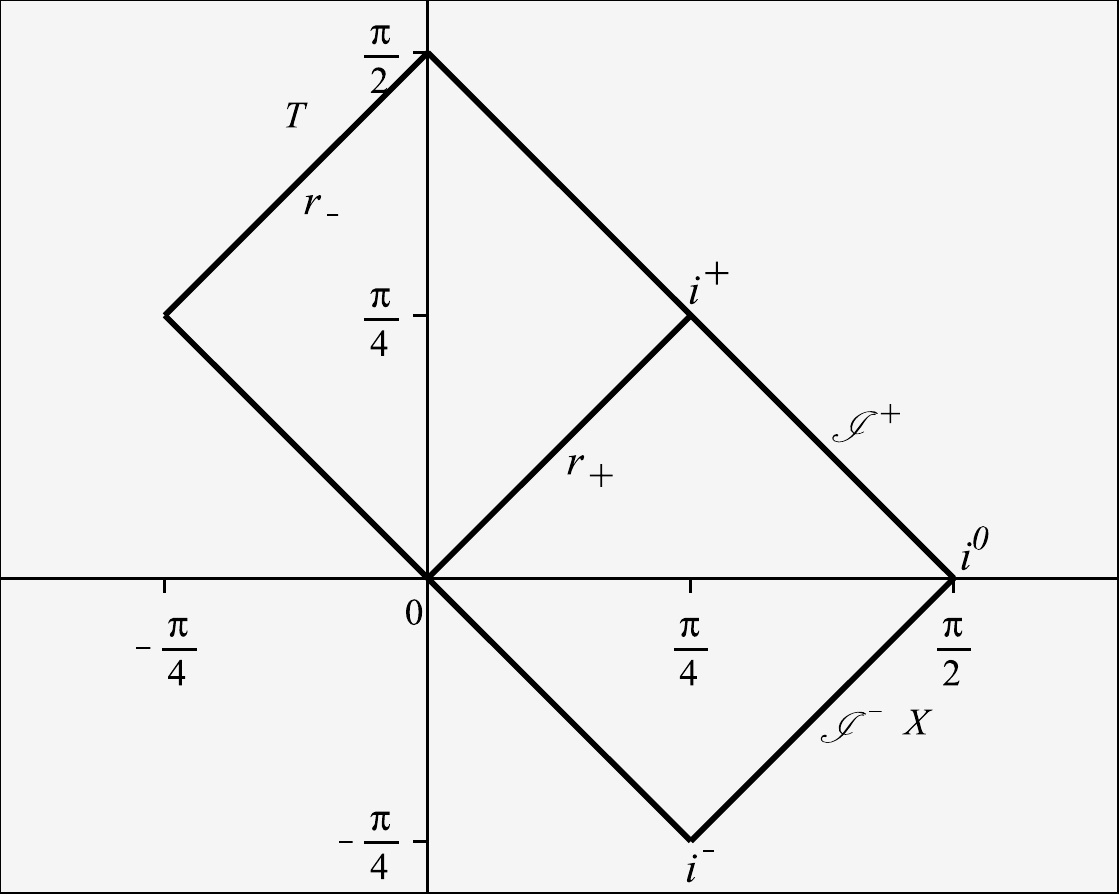}}
\hspace{15pt}%
\subfigure[][]{%
\label{CPDReissnerNordstroemF}%
\includegraphics[width=0.45\columnwidth]{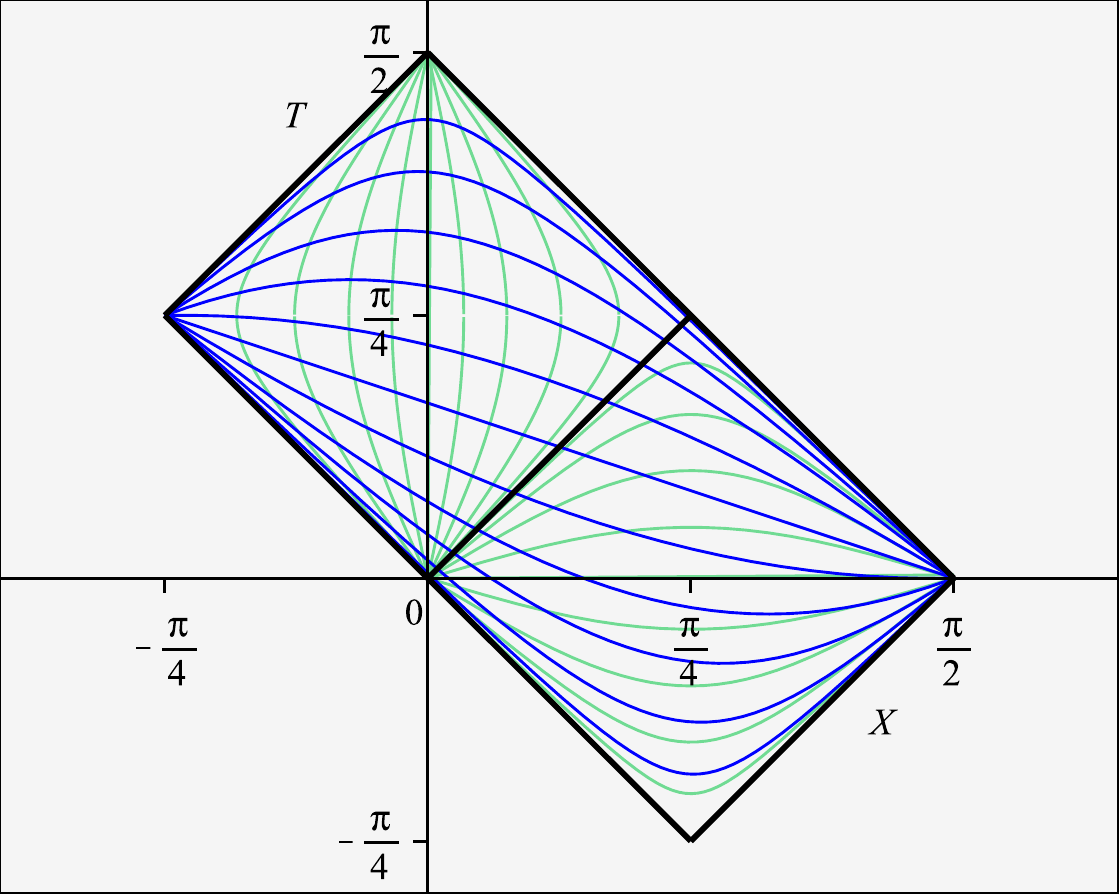}}%
\caption[...]
{Penrose diagram of the Reissner--Nordstr\"om black hole geometry up to the Cauchy horizon \subref{CPDReissnerNordstroem} and the same Penrose diagram with level sets of the Cauchy temporal function $\lambda$ defined in Equation (\ref{CauchyTempFunc2}) for values in $\pm \{0, 0.3, 0.65, 1.1, 2, 4\}$ (blue curves) and with level sets of the normalized Schwarzschild-type time coordinate $t/M \in \pm \{0, 1.24, 2.77, 4.75, 7.78\}$ for $\textnormal{B}_{\textnormal{I}}$ and $t/M \in \pm \{0, 0.86, 1.96, 3.58, 6.09\}$ for $\textnormal{B}_{\textnormal{II}}$ (aquamarine curves) for comparison \subref{CPDReissnerNordstroemF}.}%
\label{RNFull}%
\end{figure}

Next, we employ the method introduced in Section \ref{III} and work out the details of the analog of the specific integrated algebraic sigmoid function application (\ref{fssfc}) within the present framework. To this end, we have to perform the same steps as before, however, we may now omit transformation (\ref{T2}) because the Penrose diagram of the region $\textnormal{B}_{\textnormal{I}} \cup \textnormal{B}_{\textnormal{II}}$ of the Reissner--Nordstr\"om black hole geometry is already rectangularly shaped [cf.\ FIG.\ \ref{CPDReissnerNordstroem}]. This in turn leads to the first causality condition in $(\mathcal{C}4)$ assuming the form $\left|\partial_W V_{\lambda}\right| < 4/5$. Accordingly, we obtain the transformation from the above compactified Kruskal--Szekeres-type spacetime coordinates into the horizon-penetrating Cauchy coordinates 
\begin{equation*} 
\mathfrak{T}^{\textnormal{C}} \colon
\begin{cases}
\, \displaystyle \biggl(- \frac{\pi}{4}, \frac{\pi}{2}\biggr) \times \biggl(- \frac{\pi}{4}, \frac{\pi}{2}\biggr) \times (0, \pi) \times [0, 2 \pi) \rightarrow \mathbb{R} \times \biggl(- \frac{\pi}{4}, \frac{\pi}{2}\biggr) \times (0, \pi) \times [0, 2 \pi) \vspace{0.3cm} \\
\hspace{4.50cm} (T, X, \theta, \varphi) \mapsto (\lambda, X', \theta', \varphi')
\end{cases} 
\end{equation*}
with 
\begin{equation} \label{CauchyTempFunc2}
\lambda = \frac{6 T + 2 X - \pi}{\displaystyle \sqrt{10 \, \biggr[(T - X)^2 - \frac{\pi^2}{4}\biggr] \biggl[T + X - \frac{\pi}{2}\biggr] [T + X]}} \, , \quad X' = X \, , \quad \theta' = \theta \, , \quad \textnormal{and} \quad \varphi' = \varphi \, .
\end{equation}
Expressed via these coordinates, the Reissner--Nordstr\"om metric reads
\begin{equation*} 
\begin{split}
\boldsymbol{g} & = \frac{r_+ \, r_- \, (r/r_- - 1)^{1 - r^2_-/r^2_+} \, e^{- 2 \alpha r}}{\alpha^2 \bigl[\cos^2(T) - \sin^2(X)\bigr]^2 \, r^2} \, \biggl[\mathscr{G}^2 \textnormal{d}\lambda \otimes \textnormal{d}\lambda + \frac{\mathscr{E} \, \mathscr{G}}{\sqrt{1 + \mathscr{E}^2}} \, (\textnormal{d}\lambda \otimes \textnormal{d}X' + \textnormal{d}X' \otimes \textnormal{d}\lambda) - \frac{\textnormal{d}X' \otimes \textnormal{d}X'}{1 + \mathscr{E}^2}\biggr] \\ \\
& \hspace{0.4cm} - r^2 \, \boldsymbol{g}_{S^2} \, ,
\end{split}
\end{equation*}
where
\begin{equation*} 
\mathscr{E} := \frac{\sqrt{5}}{2} \, \Biggl[\lambda \biggl(X' - \frac{\pi}{8}\biggr) + \sqrt{\frac{1}{10} + \frac{\pi^2 \lambda^2}{64}} \, \Biggr] 
\end{equation*}
and
\begin{equation*} 
\mathscr{G} := \frac{1}{10 \lambda^2} \, \biggl[\biggl(\frac{1}{10} + \frac{\pi^2 \lambda^2}{64}\biggr) \bigl(1 + \mathscr{E}^2\bigr)\biggr]^{- 1/2} \, \Biggl[3 \sqrt{1 + \mathscr{E}^2} - \mathscr{E} - \sqrt{8 + \frac{5 \pi^2 \lambda^2}{4}} \, \Biggr] \, .
\end{equation*}
We emphasize that the metric coefficients $g_{\lambda \lambda}$ and $g_{\lambda X'}$ are, despite their appearance, also regular at $\lambda = 0$, which can be directly seen from the limits 
\begin{equation*} 
\lim_{|\lambda| \rightarrow 0} \mathscr{E} = \frac{1}{2 \sqrt{2}} \quad \textnormal{and} \quad \lim_{|\lambda| \rightarrow 0} \mathscr{G} = \frac{\sqrt{10}}{27} \, \biggl[4 X'^2 - \pi X' - \frac{\pi^2}{2}\biggr] \in \biggl(- \sqrt{\frac{5}{2}} \, \frac{\pi^2}{24}, 0\biggr) \, .
\end{equation*}
Therefore, this metric representation is nondegenerate everywhere on $\textnormal{B}_{\textnormal{I}} \cup \textnormal{B}_{\textnormal{II}}$. Moreover, direct computations show that the gradient of the time coordinate $\lambda$ defined in Equation (\ref{CauchyTempFunc2}) is future-directed and timelike. And by using a proof similar to the one of the Schwarzschild case (see the end of Section \ref{III}), one can demonstrate that the level sets of this time coordinate are Cauchy hypersurfaces. Thus, $\lambda$ is a Cauchy temporal function. The associated foliation of the region $\textnormal{B}_{\textnormal{I}} \cup \textnormal{B}_{\textnormal{II}}$ of the Reissner--Nordstr\"om black hole geometry is illustrated in the Penrose diagram in FIG.\ \ref{CPDReissnerNordstroemF}. We note in passing that in the Schwarzschild limit $|Q| \rightarrow 0$, some of the level sets of $\lambda$ lose their Cauchy property. This stems from the fact that all level sets located in the region above the line $3 T = - X + \pi/2$ intersect the curvature singularity of the Schwarzschild trapezoid. Hence, one obtains only a foliation of the limiting spacetime by spacelike hypersurfaces.

\section{Outlook} \label{VII}

\noindent As a future research project, we plan on generalizing our construction method of Cauchy coordinate systems to the axially symmetric Kerr black hole geometry up to the Cauchy horizon, which involves two major challenges. On the one hand, due to its nondiagonalizability, the nonextreme Kerr metric does not have the particular product structure (\ref{prodstr}), making it impossible to directly locally relate central aspects of the causal structure of the full Kerr black hole geometry to those represented in a Penrose diagram as in the present case of the spherically symmetric black hole geometries. To put it differently, since different $2$-dimensional restrictions of the Kerr black hole geometry, even when totally geodesic, lead in general to different Penrose diagrams, each depicting only slice-specific information on the global causal structure (see \cite{CrOlSz} for the examples of the axis of symmetry and the equatorial plane), there is usually no immediate connection between Cauchy surfaces in the Penrose diagram and Cauchy surfaces in the entire Kerr black hole geometry. Resolving this problem may necessitate devising a possibly higher-dimensional method specifically adapted to the metric structure of the Kerr black hole geometry. On the other hand, the time variable of the usual compactified Kruskal--Szekeres-type coordinate system for the nonextreme Kerr geometry \cite{ChandraBook, Roeken} is not a temporal function, which is in contrast to the compactified Kruskal--Szekeres time variables of the Schwarzschild and nonextreme Reissner--Nordstr\"om geometries. As this aspect is, however, paramount for the present method, we are required to first modify the construction of the Kruskal--Szekeres-type analytic extension of the nonextreme Kerr geometry accordingly. Otherwise, we could also work with an entirely different horizon-penetrating coordinate system already featuring a time coordinate that is a temporal function as basis for our geometric approach (for an example see the advanced Eddington--Finkelstein-type coordinate system used and analyzed in \cite{Roeken0}). While this may seem more suitable at first glance, the use of such a coordinate system could lead to yet unforeseen obstacles that would have to be resolved as well. In addition to this research project, we intend to apply our construction method of Cauchy coordinate systems to other spacetimes having the same metrical product structure as the Schwarzschild and Reissner--Nordstr\"om black hole geometries, thereby focusing, besides the determination of the associated Cauchy coordinate systems, on conceptual issues and the applicability of the method itself.

\section*{Acknowledgments}
\noindent The author is grateful to Miguel S\'anchez for useful discussions on the topic. This work was partially supported by the research project MTM2016-78807-C2-1-P funded by MINECO and ERDF.

\vspace{0.05cm}


\end{document}